\def\year{2017}\relax
\documentclass[letterpaper]{article}
\usepackage{aaai17}
\usepackage{times}
\usepackage{helvet}
\usepackage{courier}

\usepackage{graphicx} 

\usepackage{algorithm}
\usepackage{algorithmic}

\usepackage{amsmath}
\usepackage{amsthm}
\usepackage{amssymb}
\usepackage{amsfonts}
\usepackage{bbm}
\usepackage{color}
\usepackage[margin=10pt]{subfig}

\newtheorem{proposition}{Proposition}
\newtheorem{theorem}{Theorem}
\newtheorem{lemma}{Lemma}
\newtheorem{definition}{Definition}

\newtheorem{example}{Example}

\newcommand{\E}{\text{E}}
\newcommand{\1}{\mathbbm{1}}
\newcommand{\sat}{f,\text{sat}}
\newcommand{\unsat}{f,\text{unsat}}
\newcommand{\low}{f,\text{low}}
\newcommand{\up}{f,\text{up}}
\newcommand{\tribes}{\text{TRIBES}}
\newcommand{\one}{\boldsymbol{1}}

\newcommand\Tstrut{\rule{0pt}{2.6ex}}

\begin{document}
%
\title{General Bounds on Satisfiability Thresholds for Random CSPs via Fourier Analysis}
\author{Colin Wei \and Stefano Ermon\\
Computer Science Department\\
Stanford University\\
\texttt{\{colinwei,ermon\}@cs.stanford.edu}
}
\maketitle
\begin{abstract}
Random constraint satisfaction problems (CSPs) have been widely studied both in AI and complexity theory. Empirically and theoretically, many random CSPs have been shown to exhibit a phase transition. As the ratio of constraints to variables passes certain thresholds, they transition from being almost certainly satisfiable to unsatisfiable. The exact location of this threshold has been thoroughly investigated, but only for certain common classes of constraints.

In this paper, we present new bounds for the location of these thresholds in boolean CSPs. Our main contribution is that our bounds are fully general, and apply to any fixed constraint function that could be used to generate an ensemble of random CSPs. These bounds rely on a novel Fourier analysis and can be easily computed from the Fourier spectrum of a constraint function. Our bounds are within a constant factor of the exact threshold location for many well-studied random CSPs. We demonstrate that our bounds can be easily instantiated to obtain thresholds for many constraint functions that had not been previously studied, and evaluate them experimentally.
\end{abstract}

\section{Introduction}
Constraint satisfaction problems (CSPs) are widely used in AI, with applications in optimization, control, and planning \cite{russell2003artificial}. 
While many classes of CSPs are intractable in the worst case~\cite{cook1971complexity}, many real-world CSP instances are easy to solve in practice~\cite{vardi2014boolean}.  
As a result, there has been significant interest in understanding the average-case complexity of CSPs from multiple communities, such as AI, theoretical computer science, physics, and combinatorics~\cite{biere2009handbook}. 

Random CSPs are an important model for studying the average-case complexity of CSPs. Past works have proposed several distributional models for random CSPs \cite{molloy2003models,creignou2003generalized}. An interesting feature arising from many models is a phase transition phenomenon that occurs as one changes the ratio of number of constraints, $m$, to the number of variables, $n$. Empirical results \cite{mitchell1992hard} show that for many classes of CSPs, randomly generated instances are satisfiable with probability near 1 when $m/n$ is below a certain threshold. For $m/n$ larger than a threshold, the probability of satisfiability is close to 0. The statistical physics, computer science, and mathematics communities have focused much attention on identifying these threshold locations \cite{achlioptas2005rigorous,biere2009handbook}. 

Phase transitions for common classes of CSPs such as $k$-SAT and $k$-XORSAT are very well-studied. For $k$-SAT, researchers struggled to find tight lower bounds on the satisfiability threshold until the breakthrough work of \citeauthor{achlioptas2002asymptotic}, which provided lower bounds a constant factor away from the upper bounds. Later works closed this gap for $k$-SAT \cite{coja2013going,coja2016asymptotic}.  More recently, \citeauthor{Dudek2016CombiningTK} \shortcite{Dudek2016CombiningTK} also studied the satisfiability threshold for a more general CSP class, namely $k$-CNF-XOR, where both $k$-SAT and $k$-XORSAT constraints can be used. The results and analyses from these works, however, are all specific to the constraint classes studied.

In this paper, we provide new lower bounds on the location of the satisfiability threshold that hold for general boolean CSP classes. We focus on the setting where CSPs are generated by a single constraint type, though our analysis can extend to the setting with uniform mixtures of different constraint functions. We extend techniques from \cite{achlioptas2004threshold} and build on \cite{creignou2003generalized}, which proposes a distributional model for generating CSPs and provides lower bounds on the satisfiability threshold for these models. The significance of our work is that \textit{our bounds hold for all functions that could be used to generate random CSP instances}. The lower bounds from \cite{creignou2003generalized} are also broadly applicable, but they are looser than ours because they do not depend on constraint-specific properties. Our lower bounds are often tight (within a constant factor of upper bounds for many CSP classes) because they depend on specific properties of the Fourier spectrum of the function used to generate the random CSPs. Since these properties are simple to compute for any constraint function, our lower bounds are broadly applicable too. 

The Fourier analysis of boolean functions \cite{o2014analysis} will be vital for obtaining our main results. Expressing functions in the Fourier basis allows for clean analyses of random constraints \cite{friedgut1999sharp,barak2015beating,achim2016beyond}. Our use of Fourier analysis is inspired by the work of \citeauthor{achim2016beyond} \shortcite{achim2016beyond}, who analyze the Fourier spectra of random hash functions used as constraints in CSP-based model counting. We show that the Fourier spectrum of our constraint-generating function controls the level of spatial correlation in the set of satisfying assignments to the random CSP. If the Fourier spectrum is concentrated on first and second order coefficients (corresponding to ``low frequencies''), this correlation will be very high, roughly increasing the variance of the number of solutions to a random CSP and decreasing the probability of satisfiability. In related work, \citeauthor{montanari2011reconstruction} \shortcite{montanari2011reconstruction} also use Fourier analysis to provide tight thresholds in the case where odd Fourier coefficients are all zero. 

\section{Notation and Preliminaries}
\label{sec:notation}
In this section, we will introduce the preliminaries necessary for presenting our main theorem. First, we formally define our distribution for generating random CSPs, inspired from \cite{creignou2003generalized}. 

We will use $n$ and $m$ to denote the number of variables and number of constraints in our CSPs, respectively. We will also let $f : \{-1, 1\}^k \rightarrow \{0, 1\}$ denote a binary function and refer to $f$ as our constraint function. Often we use the term ``solution set of $f$'' to refer to the set $\{u : f(u) = 1\}$. Using the constraint function $f$, we create constraints by applying $f$ to a signed subset of $k$ variables. 

\begin{definition} [Constraint]
Let $I = (i_1, \ldots, i_k)$ be an ordered tuple of $k$ indices in $[n]$, and let $s$ be a sign vector from $\{-1, 1\}^k$. Given a vector $\sigma \in \{-1, 1\}^n$, we will define the vector $\sigma_{I, s}$ of size $k$ as follows: 
\begin{equation}
\label{eq:sigmaIs}
\sigma_{I, s} = (s_1\sigma_{i_1}, \ldots, s_k\sigma_{i_k})
\end{equation}
Now we can denote the application of $f$ to these indices by $f_{I, s}(\sigma) = f(\sigma_{I, s})$. We call $f_{I, s}$ a constraint, and we say that $\sigma \in \{-1, 1\}^n$ satisfies the constraint $f_{I, s}$ if $f_{I, s}(\sigma) = 1$.
\end{definition}
The definition of a CSP generated from $f$ follows. 
\begin{definition} [CSP generated from $f$]
We will represent a CSP with $m$ constraints and $n$ variables generated from $f$ as a collection of constraints $C_f(n, m) = \{f_{I_1, s_1}, \ldots, f_{I_m, s_m}\}$. Then $\sigma \in \{-1, 1\}^n$ satisfies $C_f(n, m)$ if $\sigma$ satisfies $f_{I_j, s_j}$ for $j = 1, \ldots, m$.
\end{definition}
\begin{example} [3-SAT]
Let $f : \{-1, 1\}^3 \rightarrow \{1, 0\}$ where $f(u) = 0$ for $u = (-1, -1, -1)$ and $f(u) = 1$ for all other $u$. Then $f$ is the constraint function for the 3-SAT problem.
\end{example}
With these basic definitions in place, we are ready to introduce the model for random CSPs. 
\subsection{Random CSPs}
\label{sec:randcsp}
We discuss our model for randomly generating  $C_f(n, m)$, and formally define a ``satisfiability threshold.''

To generate instances of $C_f(n, m)$, we simply choose $I_1, \ldots, I_m$ and $s_1, \ldots, s_m$ uniformly at random. For completeness, we sample without replacement, i.e. there are no repeated variables in a constraint, and no duplicate constraints. However, sampling with replacement does not affect final results. For the rest of this paper we abuse notation and let $C_f(n, m)$ denote a randomly generated CSP instance following this model. 

Now we formally discuss satisfiability thresholds. We let $r = m/n$. For many constraint functions $f$, there exist thresholds $r_{\sat}$ and $r_{\unsat}$ such that 
\begin{align*}
\lim_{n \rightarrow \infty} \Pr[C_f(n, rn) \ \text{is satisfiable}] = 
\begin{cases}
1 \ \text{if} &\ r < r_{\sat}\\
0 \ \text{if} &\ r > r_{\unsat}
\end{cases}
\end{align*}
In general, it is unknown whether $r_{\sat} = r_{\unsat}$, but for some problems such as $k$-SAT (for large $k$) and $k$-XORSAT, affirmative results exist \cite{ding2015proof,pittel2016satisfiability}. If $r_{\sat} = r_{\unsat}$, then we say that the random CSP $C_f(n, m)$ exhibits a sharp threshold in $m/n$. 

We are concerned with finding lower bounds $r_{\low}$ on $r_{\unsat}$ such that there exists a constant $C > 0$ independent of $n$ so that for sufficiently large $n$, 
\begin{equation}
\label{eq:rflower}
\Pr[C_f(n, rn) \ \text{is satisfiable}] > C \ \text{for all} \ r < r_{\low}
\end{equation}
For CSP classes with a sharp threshold, \eqref{eq:rflower} implies that 
\begin{align*}
\lim_{n \rightarrow \infty} \Pr[C_f(n, rn) \ \text{is satisfiable}] = 1 \ \text{for all} \ r < r_{\low}
\end{align*}
We also wish to find upper bounds $r_{\up}$ such that 
\begin{align*}
\lim_{n \rightarrow \infty} \Pr[C_f(n, rn) \ \text{is satisfiable}] = 0 \ \text{for all} \ r > r_{\up}
\end{align*}
For an example of these quantities instantiated on a concrete example, refer to the experiments in Section \ref{sec:experiments}.

We provide a value for $r_{\up}$ which was derived earlier in \cite{creignou2003generalized}. \citeauthor{dubois2001upper} \shortcite{dubois2001upper} and \citeauthor{creignou2007expected} \shortcite{creignou2007expected} provide methods for obtaining tighter upper bounds, but the looser values that we use are sufficient for showing that $r_{\low}$ is on the same asymptotic order as $r_{\unsat}$ for many choices of $f$. 

The bound $r_{\low}$ depends on both the symmetry and size of the solution set of $f$. The more assignments $u \in \{-1, 1\}^k$ such that $f(u) = 1$, the more likely it is that each constraint is satisfied. Increased symmetry reduces the variance in the number of solutions to $C_{f}(n, rn)$, so solutions are more spread out among possible CSPs in our class and the probability that $C_f(n, rn)$ will have a solution is higher. We will formally quantify this symmetry in terms of the Fourier spectrum of $f$, which we introduce next. 

\subsection{Fourier Expansion of Boolean Functions}
We discuss basics of Fourier analysis of boolean functions. For a detailed review, refer to \cite{o2014analysis}. We define the vector space $\mathcal{F}_k$ of all functions mapping $\{-1, 1\}^k$ to $\mathbb{R}$. The set $\mathcal{F}_k$ has the inner product $\langle f_1, f_2 \rangle = \sum_{u \in \{-1, 1\}^k} f_1(u)f_2(u)/2^k$ for any $f_1, f_2 \in \mathcal{F}_k$.

This inner product space has orthonormal basis vectors $\chi_S$, where the parity functions $\chi_S$ follow $\chi_S(u) = \prod_{i \in S} u_i$ for all $S \subseteq [k]$, subsets of the $k$ indices. Because $(\chi_S)_{S \subseteq [k]}$ forms an orthonormal basis, if we write 
\begin{equation}
\label{eq:hatfS}
\hat{f}(S) = \langle f, \chi_S \rangle = \frac{1}{2^k} \sum_{u \in \{-1, 1\}^k} f(u)\chi_S(u)
\end{equation}
then we can write $f$ as a linear combination of these vectors: $f = \sum_{S \subseteq [k]} \hat{f}(S) \chi_S$. We note that when $S = \emptyset$, the empty set, $\hat{f}(\emptyset)$ is simply the average of $f$ over $\{-1, 1\}^k$. We will refer to the coefficients $(\hat{f}(S))_{S \subseteq [k]}$ as the Fourier spectrum of $f$. Since these coefficients are well-studied in theoretical computer science \cite{o2014analysis}, the Fourier coefficients of many boolean functions are easily obtained.

\begin{example} [3-SAT] 
\label{ex:3satfourier}
For $3$-SAT, $\hat{f}(\emptyset) = 7/8$. $\hat{f}(\{1\}) = \hat{f}(\{2\}) = \hat{f}(\{3\}) = 1/8$, and $\hat{f}(\{1, 2\}) = \hat{f}(\{2, 3\}) = \hat{f}(\{1, 3\}) = -1/8$.
\end{example}

Representing $f$ in the Fourier bases will facilitate our proofs, providing a simple way to express expectations over our random CSPs. The Fourier spectrum can also provide a measure of ``symmetry'' in $f$ - if some values of $\hat{f}(S)$ are high where $|S| = 1$, then satisfying assignments to $f$ are more skewed in the variable corresponding to $S$. We will show how this impacts satisfiability in Section \ref{sec:boundexplained}.

\begin{figure*}
\centering
\begin{tabular}{|c | c | c| c|}
\hline
\Tstrut
CSP class ($f$) & Best lower bound on $r_{\unsat}$ & Our bound $r_{\low}$ & Upper bound $r_{\up}$\\ \hline
\Tstrut
$k$-XORSAT & 1 & $\frac{1}{2}$ & 1\\
\Tstrut
$k$-SAT & $2^k\ln 2 - \frac{1 + \ln 2}{2} - o_k(1)$ & $2^{k - 1} - O(k)$ & $2^k \ln 2$\\ 
\Tstrut
$k$-NAESAT & $2^{k - 1} \ln 2 - \frac{\ln 2}{2} - \frac{1}{4} - o_k(1)$ & $2^{k - 2} - \frac{1}{2}$ & $2^{k - 1}\ln 2$ \\
\Tstrut
$k$-MAJORITY & ? & $\frac{\frac{1}{2} - k\binom{k - 1} {\frac{k - 1}{2}}^2 2^{-2k +1}}{1 + k\binom{k - 1}{\frac{k - 1}{2}}^2 2^{-2k+ 2}} = 0.111 - o_a(1)$ & $1$  \\
\Tstrut
$a$-MAJ $\otimes 3$-MAJ & ? & $\frac{\frac{1}{2} - 3a\binom{a - 1}{\frac{a - 1}{2}}^2 2^{-2a - 1}}{1 + 3a\binom{a - 1}{\frac{a - 1}{2}}^2 2^{-2a - 2}} = 0.177 - o_a(1)$ & $1$\\
\Tstrut
$k$-MOD-3 & ? & $\frac{1}{4} - o_k(1)$ & $\frac{\ln 2}{\ln 3} + o_k(1)$\\
\Tstrut
OR$_b \otimes$ XOR$_a$ & ? & $2^{b - 1} - 1/2$ & $2^{b - 1} \ln 2$\\
\hline
\end{tabular}
\caption{We compare the best known lower bounds on the satisfiability threshold to our lower and upper bounds. For $k$-XORSAT \cite{pittel2016satisfiability}, $k$-SAT \cite{ding2015proof}, and $k$-NAESAT \cite{coja2012catching}, the numbers listed are known as exact sharp threshold locations. For the last four, we do not know of existing lower bounds. $\otimes$ is the composition operator for boolean functions, and we define these functions in Section \ref{sec:figfunc}.}
\label{fig:bounds}
\end{figure*}

\section{Main Results}
\label{sec:mainresult}
We provide simple formula for $r_{\up}$. A similar result is in \cite{creignou2003generalized}, and the full proof is in the appendix.
\begin{proposition}
\label{prop:upper}
For all constraint functions $f$, let 
\begin{align*}
r_{\up} = \frac{\log 2}{\log 1/\hat{f}(\emptyset)} < \frac{\log 2}{1 - \hat{f}(\emptyset)}
\end{align*}
If $r \ge r_{\up}$, $\lim_{n \rightarrow \infty} \Pr[C_f(n, rn) \ \text{is satisfiable}] = 0$. 
\end{proposition}
\begin{proof} [Proof Sketch]
We compute the expected solution count for $C_f(n, rn)$. The expected solution count will scale with $\hat{f}(\emptyset)$, since  $2^k \hat{f}(\emptyset)$ is simply the number of $u \in \{-1, 1\}^k$ where $f(u) = 1$ and therefore governs how easily each constraint will be satisfied. If $r > r_{\up}$, the expected solution count converges to 0 as $n \rightarrow \infty$, so Markov's inequality implies that the probability that a solution exists goes to 0.
\end{proof}

Next, we will present our value for $r_{\low}$. First, some notation: let $U = \{u \in \{-1, 1\}^k : f(u) = 1\}$, and let $A$ be the $k \times |U|$ matrix whose columns are the elements of $u$. We will use $A^+$ to denote the Moore-Penrose pseudoinverse of $A$. For a reference on this, see \cite{barata2012moore}. Finally, let $\one$ be the $|U|$-dimensional vector of 1's.
\begin{example} [$3$-SAT] For $3$-SAT, $k = 3$ and $|U| = 7$, and we can write $A$ as follows (up to permutation of its columns):
\begin{align*}
A = 
\left[
\begin{array} {r r r r r r r}
1 & 1 & 1 & 1 & -1 & -1 & -1\\
1 & 1 & -1 & -1 & 1 & 1 & -1\\
1 & -1 & 1 & -1 & 1 & -1 & 1
\end{array}
\right]
\end{align*}
where columns of $A$ satisfy the $3$-SAT constraint function.
\end{example}
The following main theorem provides the \textit{first computable equation for obtaining lower bounds that are specific to the constraint-generating function $f$.} 

\begin{theorem}
\label{thm:main}
For all constraint functions $f$, let 
\begin{align*}
r_{\low} = \frac{1}{2} \frac{c}{1 - c} \text{ where } c = \hat{f}(\emptyset) - \frac{\one^T A^+ A \one}{2^k}
\end{align*}
If $r < r_{\low}$, then there exists a constant $C>0$ such that $\lim_{n \rightarrow \infty} \Pr[C_f(n, rn) \ \text{is satisfiable}] > C$. 
\end{theorem}

The lower bound $r_{\low}$ is an increasing function of $c$ which is dependent on two quantities. First, with higher values of $\hat{f}(\emptyset)$, $C_f(n, rn)$ will have more satisfying assignments on average, so $c$ and the threshold value will be higher. Second, $c$ depends on the level of symmetry in the solution set of $f$, which we will show is connected to the Fourier spectrum of $f$. We explain this dependence in Section \ref{sec:boundexplained}. 

In comparison, \citeauthor{creignou2003generalized} \shortcite{creignou2003generalized} obtain lower bounds which depend only on the arity of $f$. While we cannot make an exact comparison because \citeauthor{creignou2003generalized} use a different random CSP ensemble, for reference, they provide the general lower bound of $1/(ke^k - k)$ expected constraints per variable for functions of arity $k$. Our bounds are much tighter because of their specificity while remaining simple to compute. To demonstrate, we instantiate our bounds for some example constraint functions in Figure \ref{fig:bounds}. Whereas their bounds are exponentially decreasing in $k$, our bounds are constant or increasing in $k$ for the functions shown. 

\subsection{Constraint Functions in Figure \ref{fig:bounds}}
\label{sec:figfunc}
We define the constraint functions in Figure \ref{fig:bounds}. Unless specified otherwise, they will be in the form $f : \{-1, 1\}^{k} \rightarrow \{0, 1\}$. 
\begin{enumerate}
\item $k$-SAT: $f(u) = 0$ if $u$ is the all negative ones vector, and $f(u) = 1$ otherwise.
\item $k$-XORSAT: $f(u) = \1(\chi_{[k]}(u) = -1)$
\item $k$-NAESAT: $f(u) = 0$ if $u$ is the all negative ones or all ones vector, and $f(u) = 1$ otherwise.
\item $k$-MAJORITY: Defined when $k$ is odd, $f(u) = 1$ if more than half of the variables of $u$ are $1$.
\item $a$-MAJ $\otimes$ $3$-MAJ: Defined when $a$ is odd, where $f : \{-1, 1\}^{3a} \rightarrow \{0, 1\}$. Defined as the composition of $a$-MAJORITY on $a$ groups of $3$-MAJORITY, as follows: 
\begin{align*}
f(u_1, \ldots, u_{3a}) =\\ f_{a-\text{MAJ}}(f_{3-\text{MAJ}}(u_1, u_2, u_3), \ldots, \\ f_{3-\text{MAJ}}(u_{3a - 2}, u_{3a -1}, u_{3a}))
\end{align*}
\item $k$-MOD-3: $f(u) = 1$ when the number of 1's in $u$ is divisible by $3$, and 0 otherwise.
\item OR$_b \otimes$ XOR$_a$: In this case, $f : \{-1, 1\}^{ab} \rightarrow \{0, 1\}$, and $f$ is the composition of a OR over $b$ groups of XORs over $a$ variables, as follows: 
\begin{align*}
f(u_{1}, \ldots, u_{ab}) =\\
f_{\text{OR}_b}(f_{\text{XOR}_a}(u_1, \ldots, u_a), \ldots,\\
f_{\text{XOR}_{a}}(u_{ab - a + 1}, \ldots, u_{ab}))
\end{align*}
\end{enumerate}

While the last four constraint functions have not been analyzed much in the existing CSP literature, these types of general constraints are of practical interest because of \cite{achim2016beyond}, which performs probabilistic inference by solving CSPs based on arbitrary hash functions. For example, \citeauthor{achim2016beyond} \shortcite{achim2016beyond} show that MAJORITY constraints are effective in practice for solving probabilistic inference problems. 
\subsection{Connecting Bounds with Fourier Spectrum}
\label{sec:boundexplained}
We explain how the Fourier spectrum can help us interpret Theorem \ref{thm:main}. We first show the connection between $c$ and the Fourier spectrum. 
Let $\hat{f}_{S : |S| = 1}$ be the $k$-dimensional vector whose entries are Fourier coefficients of $f$ for size 1 sets. Let $B$ be the $k \times k$ matrix with diagonal entries $B_{ii} = \hat{f}(\emptyset)$ and off-diagonal entries $B_{ij} = \hat{f}(\{i, j\})$ for $i \ne j$. 
\begin{example} [3-SAT]
Following the coefficients in Example \ref{ex:3satfourier}, for $3$-SAT, $\hat{f}_{S : |S| = 1} = (1/8, 1/8, 1/8)$ and 
\begin{align*}
B = \left[
\begin{array} {r r r}
7/8 & -1/8 & -1/8\\
-1/8 & 7/8 & -1/8\\
-1/8 & -1/8 & 7/8
\end{array}
\right]
\end{align*}
\end{example}
\begin{lemma}
\label{lem:fourierc}
When the rows of $A$ are linearly independent, $c = \hat{f}(\emptyset) - \hat{f}_{S : |S| = 1}^T B^{-1} \hat{f}_{S : |S| = 1}$.
\end{lemma}

From this lemma, we see that larger values of $c$ correspond to smaller $\hat{f}_{S : |S| = 1}$. These terms will measure the amount of ``symmetry'' in the solution set for $f$. 
The matrix $B$ and vectors $\hat{f}_{S : |S| = 1}$ are easily obtained for many $f$ since Fourier coefficients are well-studied \cite{o2014analysis}.  

Figure \ref{fig:bounds} shows how $\hat{f}(\emptyset) - c$ and $f$ relate. Since $k$-SAT has a mostly symmetric solution set, $\hat{f}_{k-\text{SAT}}(\emptyset) - c = O(k/2^{2k})$ since $c = 1 - 2^{-k} - O(k/2^{2k})$, which is small compared to $\hat{f}_{k-\text{SAT}}(\emptyset)$. The solution set of $k$-NAESAT is completely symmetric as if $f_{k-\text{NAESAT}}(x) = 1$, then $f_{k-\text{NAESAT}}(-x) = 1$. Thus, $f_{k-\text{NAESAT}}$ has 0 weight on Fourier coefficients for sets with odd size so we can compute that $\hat{f}_{k-\text{NAESAT}}(\emptyset) - c = 0$. $k$-MAJORITY, however, is less symmetric, as shown by larger first order coefficients. Here, the bound in Figure \ref{fig:bounds} gives 
$\lim_{k \rightarrow \infty} \hat{f}_{k-\text{MAJORITY}}(\emptyset) - c = 1/\pi$, which is large compared to $\hat{f}_{k-\text{MAJORITY}}(\emptyset) \approx 1/2$.

\section{Proof Strategy}
Our proof relies on the second moment method, which has been applied with great success to achieve lower bounds for problems such as $k$-SAT \cite{achlioptas2004threshold} and $k$-XORSAT \cite{dubois20023}. The second moment method is based on the following lemma, which can be derived using the Cauchy-Schwarz inequality: 
\begin{lemma}
\label{lem:secondmomentmethod}
Let $X$ be any real-valued random variable. Then
\begin {equation}
\label{eq:secondmom}
Pr[X \ne 0] = \Pr[|X| \ne 0] \ge \frac{\E[|X|]^2}{\E[X^2]} \ge \frac{\E[X]^2}{\E[X^2]}
\end {equation}
\end{lemma}
If $X$ is only nonzero when $C_f(n, rn)$ has a solution, we obtain lower bounds on the probability that a solution exists by upper bounding $\E[X^2]$. For example, we could let $X$ be the number of solutions to $C_f(n, rn)$. However, as shown in \cite{achlioptas2004threshold}, this choice of $X$ fails in most cases. Whether two different assignments satisfy $C_f(n, rn)$ is correlated: if the assignments are close in Hamming distance and one assignment is satisfying, it is more likely that the other is satisfying as well. This will make $\E[X^2]$ much larger than $\E[X]^2$, so \eqref{eq:secondmom} will not provide useful information. Figure \ref{fig:failure} demonstrates this failure for $k$-SAT. \citeauthor{achlioptas2004threshold} \shortcite{achlioptas2004threshold} show formally that the ratio $\E[X]^2/\E[X^2]$ will decrease exponentially (albeit at a slow rate). On the other hand, $k$-NAESAT is ``symmetric'', so the second moment method works directly here. In the plot, $\E[X]^2/\E[X^2]$ for $3$-NAESAT stays above a constant. This also follows formally from our main theorem as well as \cite{achlioptas2002asymptotic}. We formally define our requirements on symmetry in \eqref{eq:derivhalf0}.  

We circumvent this issue by weighting solutions to reduce correlations before applying the second moment method. As in \cite{achlioptas2004threshold}, we use a weighting which factors over constraints in $C_f(n, rn)$ and apply the second moment method to the random variable 
\begin{equation}
\label{eq:X}
X = \sum_{\sigma \in \{-1, 1\}^n} \prod_{c \in C_f(n, rn)}w(\sigma, c)
\end{equation}
where $C_f(n, rn)$ is a collection of constraints $\{f_{I_1, s_1}, \ldots, f_{I_m, s_m}\}$ and the randomness in $X$ comes over the choices of $I_j, s_j$. Now we can restrict our attention to constraint weightings of the form $w(\sigma, f_{I, s}) = w(\sigma_{I, s})$. In the special case where $w(\sigma_{I, s}) = f(\sigma_{I, s})$, $X$ will simply represent the number of solutions to $C_f(n, rn)$. In general, we require $w(\sigma_{I, s}) = 0$ whenever $f(\sigma_{I, s}) = 0$. This way, if $X \ne 0$, then $C_f(n, rn)$ must have a solution.  

For convenience, we assume that the index sets $I_1, \ldots, I_m$ are sampled with replacement. They are chosen uniformly from $[n]^k$. We also allow constraints to be identical. In the appendix, we justify why proofs in this setting carry over to the without-replacement setting in Section \ref{sec:randcsp} and also provide full proofs to the lemmas presented below.  

\begin{figure*}
\centering
\subfloat[$\E{[X]}^2/\E{[X^2]}$ vs. $n$ for $3$-SAT and $3$-NAESAT. $X$ is the solution count, $r = 1$.]
{\includegraphics[width=0.32\textwidth]{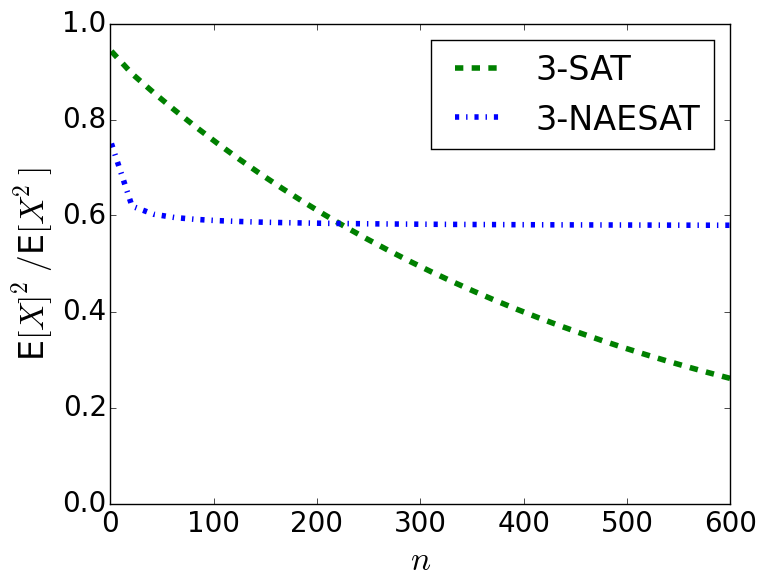}
\label{fig:failure}
}
\subfloat[$g_w(\alpha)$ when $f$ is $k$-XORSAT and $w = f$.]{\includegraphics[width=0.32\textwidth]{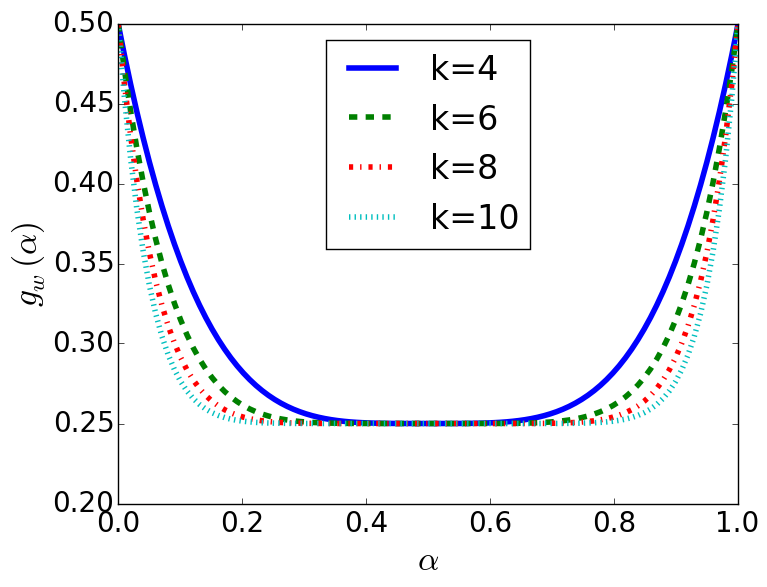}
\label{fig:g_w}}
\subfloat[$\psi_r(\alpha)$ with $f$ as $5$-NAESAT and $w = f$.]{\includegraphics[width=0.32\textwidth]{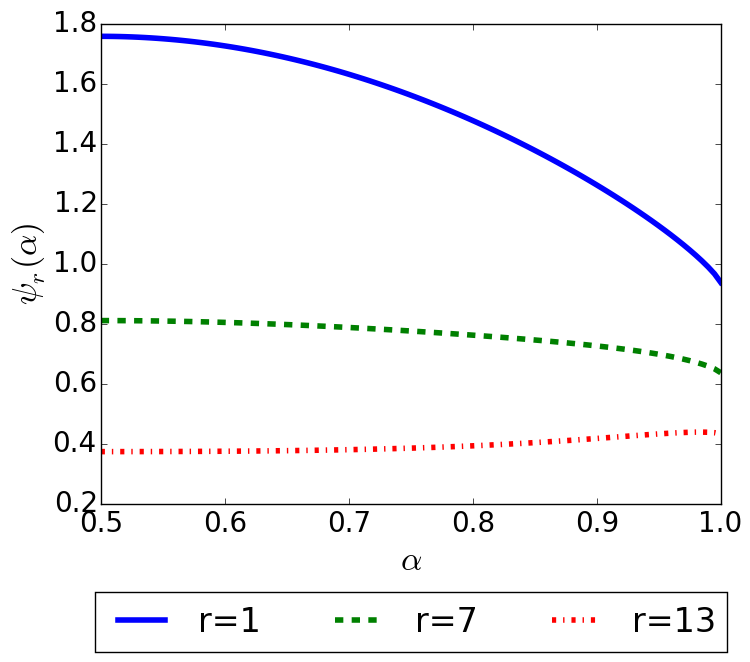}
\label{fig:psi}}
\caption{For concreteness, we provide sample plots of the relevant quantities in our proofs.}
\end{figure*}

In this setting, we will compute the first and second moments of the $X$ chosen in \eqref{eq:X} in terms of the Fourier spectrum of $w$.
\begin{lemma}
\label{lem:firstmomentsq}
The squared first moment of $X$ is given by 
\begin{equation}
\label{eq:E[X]sq}
\E[X]^2 = 2^{2n}(\hat{w}(\emptyset)^2)^{rn}
\end{equation}
\end{lemma}
\begin{proof}
We can expand $\E[X]$ as follows: 
\begin{align}
\E[X] &= \sum_{\sigma \in \{-1, 1\}^n} \E\left[\prod_{j = 1}^{rn}w(\sigma_{I_j, s_j})\right] \notag \\
&= \sum_{\sigma \in \{-1, 1\}^n} \E[w(\sigma_{I, s})]^{rn} \label{eq:EXint}
\end{align}
where we used the fact that constraints are chosen independently. Now we claim that for any $u \in \{-1, 1\}^k$, 
\begin{align*}
\Pr[\sigma_{I, s} = u] = \frac{1}{2^k}
\end{align*}
This follows from the fact that we choose $s$ uniformly over $\{-1, 1\}^k$ and our definition of $\sigma_{I, s}$ in \eqref{eq:sigmaIs}. Thus, 
\begin{align*}
\E[w(\sigma_{I, s})] &= \sum_{u \in \{-1, 1\}^k} w(u)\Pr[\sigma_{I, s} = u]\\
&= \frac{1}{2^k}\sum_{u \in \{-1, 1\}^k} w(u)\\
&= \hat{w}(\emptyset)
\end{align*}
Plugging back into \eqref{eq:EXint} gives the desired result. 
\end{proof}
Next, we will compute the second moment $\E[X^2]$.
\begin{lemma}
\label{lem:secmoment}
Let $g_w(\alpha) = \sum_{S \subseteq [k]} (2\alpha - 1)^{|S|} \hat{w}(S)^2$. The second moment of $X$ is given by
\begin{equation}
\label{eq:E[Xsq]}
\E[X^2] = 2^n \sum_{j = 0}^n {n \choose j} g_w(j/n)^{rn}
\end{equation}
\end{lemma}

The function $g_w(\alpha)$ is similar to the noise sensitivity of a boolean function \cite{o2003computational} and measures the correlation in the value of $w$ between two assignments $\sigma$, $\tau$ which overlap at $\alpha(\sigma, \tau) n$ locations. As a visual example, Figure \ref{fig:g_w} shows how $g_w(\alpha)$ changes for $k$-XORSAT with varying $k$. The key of our proof is showing that $\E[w(\sigma_{I, s})w(\tau_{I, s})] = g_w(\alpha(\sigma, \tau))$ for a random constraint $f_{I, s}$.  

We will now write $\E[X]^2$ in terms of $g_w$. Since $g_w(1/2) = \hat{w}(\emptyset)^2$, plugging this into \eqref{eq:E[X]sq} gives us
\begin{equation}
\label{eq:e[X]sqing}
\E[X]^2 = 2^{2n} g_w(1/2)^{rn} 
\end{equation}
This motivates us to apply the following lemma from \cite{achlioptas2004threshold}, which will allow us to translate bounds on $g_w(\alpha)$ into bounds on $\E[X]^2/\E[X^2]$: 
\begin{lemma}
\label{lem:achlioptaslem}
Let $\phi$ be any real, positive, twice-differentiable function on $[0, 1]$ and let 
\begin{align*}
S_n = \sum_{j = 0}^{n} {n \choose j} \phi(j/n)^n
\end{align*}
Define $\psi$ on $[0, 1]$ as $\psi(\alpha) = \frac{\phi(\alpha)}{\alpha^{\alpha}(1 - \alpha)^{1 - \alpha}}$.
If there exists $\alpha_{\max} \in (0, 1)$ such that $\psi(\alpha_{\max}) > \psi(\alpha)$ for all $\alpha \ne \alpha_{\max}$, and $\psi''(\alpha_{\max}) < 0$, then there exist constants $B, C > 0$ such that for all sufficiently large $n$, 
\begin{equation}
\label{eq:lemachlioptasres}
B \psi(\alpha_{\max})^n \le S_n \le C \psi(\alpha_{\max})^n 
\end{equation}
\end{lemma}
To apply the lemma, we can define $\phi_r(\alpha) = g_w(\alpha)^r$ and $\psi_r(\alpha) = \frac{\phi_r(\alpha)}{\alpha^{\alpha}(1 - \alpha)^{1 - \alpha}}$. Then from \eqref{eq:e[X]sqing}, we note that 
\begin{equation}
\psi_r(1/2)^n = 2^n(g_w(1/2)^r)^n = \E[X]^2/2^n
\label{eq:psiE[X]sq}
\end{equation}
On the other hand, from \eqref{eq:E[Xsq]},
\begin{equation}
\sum_{j = 0}^n {n \choose j} \phi_r(j/n)^n = \E[X^2]/2^n
\label{eq:psiE[Xsq]}
\end{equation}
so if the conditions of Lemma \ref{lem:achlioptaslem} hold for $\alpha_{\max} = 1/2$, we recover that $\E[X]^2/\E[X^2] \ge C$ for some constant $C > 0$. 

One requirement for $\psi_r(\alpha)$ to be maximized at $\alpha = 1/2$ is that $\psi_r'(1/2) = 0$. Expanding $\psi_r'(1/2)$ gives $2g_w(1/2)^{r - 1}(rg_w'(1/2)) = 0$.
Since $g_w(1/2) = \hat{w}(\emptyset)^2 > 0$, we thus require 
\begin{equation}
\label{eq:derivhalf0}
g_w'(1/2) = 2\sum_{S \subseteq [k] : |S| = 1} \hat{w}(S)^2 = 0
\end{equation}
In order to satisfy \eqref{eq:derivhalf0}, we need $\hat{w}(S) = 0$ for all $S \subseteq [k]$ where $|S| = 1$. To use Lemma \ref{lem:achlioptaslem}, we would like to choose $w$ such that \eqref{eq:derivhalf0} holds. We discuss how to choose $w$ to optimize our lower bounds in the appendix. In the next section, we will provide $r$ so that the conditions of Lemma \ref{lem:achlioptaslem} hold at $\alpha = 1/2$ for arbitrary $w$ when \eqref{eq:derivhalf0} is satisfied. 

\subsection{Bounding the Second Moment For Fixed $w$}
We give a general bound on $r$ in terms of our weight function $w$ so that the conditions of Lemma \ref{lem:achlioptaslem} are satisfied for $\alpha = 1/2$. For now, the only constraint we place on $w$ is that \eqref{eq:derivhalf0} holds. The next lemma lets us consider only $\alpha \in [1/2, 1]$.  
\begin{lemma}
\label{lem:alphagehalf}
Let $\alpha \ge 1/2$. Then $g_w(\alpha) \ge g_w(1 - \alpha)$. 
\end{lemma}
This lemma follows because $g_w(\alpha)$ is a polynomial in $(2\alpha - 1)$ with nonnegative coefficients, and $(2\alpha - 1) > 0$ for $\alpha > 1/2$.  

Now we can bound $\psi_r(\alpha)$ for $\alpha \in [1/2, 1]$. Combined with Lemma \ref{lem:alphagehalf}, the next lemma will give conditions on $r$ such that $\psi_r(1/2) > \psi_r(\alpha)$ for all $\alpha \in [0, 1]$.

\begin{lemma}
\label{lem:rbound}
Let the weight function $w$ satisfy \eqref{eq:derivhalf0}. If 
\begin{equation}
\label{eq:rcond}
r \le \frac{1}{2}\frac{\hat{w}(\emptyset)^2}{\sum_{S : |S| \ge 2} \hat{w}(S)^2}
\end{equation}
$\psi_r(1/2) > \psi_r(\alpha)$ for $\alpha \in [0, 1]$ and $\psi_r''(1/2) < 0$.   
\end{lemma}
Figure \ref{fig:psi} shows how $r$ controls the shape of the function $\psi_r(\alpha)$. As $r$ increases, $\psi_r''(1/2)$ becomes positive and $\psi_r(\alpha)$ will no longer attain a local maximum in that region. The key step in proving Lemma \ref{lem:rbound} is rearranging $\psi_r(1/2) > \psi_r(\alpha)$ and simplify calculations by using approximations for the logarithmic terms that appear. 

Our bound on $r$ compares the average of $w$ over $\{-1, 1\}^k$ with the correlations between $w$ and the Fourier basis functions. If $w$ has strong correlations with the other Fourier basis functions, two assignments which are equal at $\alpha n$ variables will likely either be both satisfying or both not satisfying as $\alpha$ approaches $1$. This increases $\E[X^2]$ but not $\E[X]^2$ and makes Lemma \ref{lem:secondmomentmethod} provide a trivial bound if $r$ is too large. Thus, if $w$ has strong correlations with the Fourier basis functions, we must choose smaller $r$ as reflected by \eqref{eq:rcond}. 

To get the tightest bounds, we wish to maximize the expression in \eqref{eq:rcond}. Although we prove our lemma for general $w$ requiring only \eqref{eq:derivhalf0}, we also need $w(u) = 0$ whenever $f(u) = 0$ to apply our lemma to satisfiability. Recalling our definition of $X$ in \eqref{eq:X}, this condition ensures that $C_f(n, rn)$ has a solution whenever $X \ne 0$. Thus,  
\begin{equation}
\label{eq:wlambdaf}
w(u) = \lambda(u) f(u)
\end{equation} 
for some $\lambda : \{-1, 1\}^k \rightarrow \mathbb{R}$. If we disregard \eqref{eq:derivhalf0}, choosing $\lambda(u) = 1$ would maximize the bound on $r$ in \eqref{eq:rcond}. The additional requirement of \eqref{eq:derivhalf0} for the second moment method to succeed can be viewed as a ``symmetrization penalty'' on $r$. In the appendix, we discuss how to choose $w$ to optimize our bound on $r$ while satisfying \eqref{eq:derivhalf0} and \eqref{eq:wlambdaf}.

\begin{figure*}
\centering
\subfloat[$a=2,b=4$]{
\includegraphics[scale=0.32]{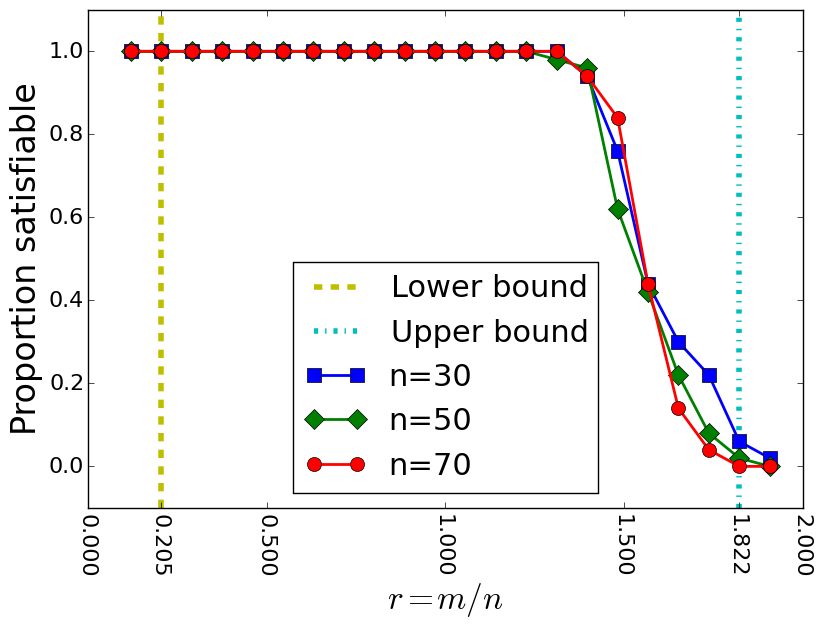}}
\subfloat[$a=2,b=5$]{
\includegraphics[scale=0.32]{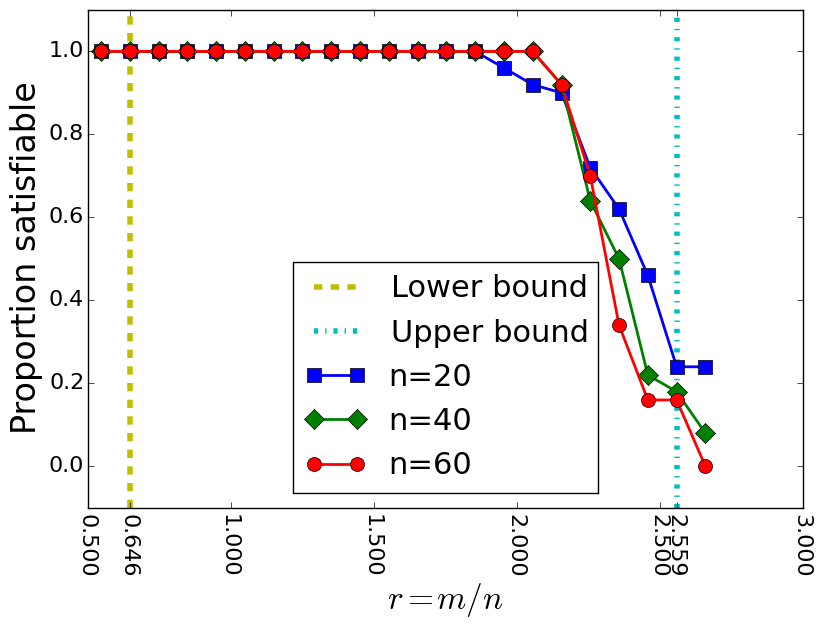}}
\caption{Proportion of CSPs satisfiable out of 50 trials vs. $r$ for tribes functions. We show our bounds for reference.}
\label{fig:experiments}
\end{figure*}

\subsection{Proving the Main Theorem}
\label{sec:mainproof}
We will combine our lemmas to prove Theorem \ref{thm:main}. 
\begin{proof} [Proof of Theorem \ref{thm:main}] 
We wish to apply the second moment method on $X$ defined in \eqref{eq:X}, where $w$ is a function we use to weigh assignments to individual constraints. We choose $w$ as described in the full version of the paper, which satisfies both \eqref{eq:derivhalf0} and \eqref{eq:wlambdaf}. Since $w$ satisfies \eqref{eq:wlambdaf}, $\Pr[C_f(n, rn) \ \text{is satisfiable}] \ge \Pr[X \ne 0] \ge \E[X]^2/\E[X^2]$
by Lemma \ref{lem:secondmomentmethod}. Now we will use Lemma \ref{lem:rbound} to show that the conditions for Lemma \ref{lem:achlioptaslem} are satisfied for $\phi_r = g_w(\alpha)^r$ and $r$ satisfying \eqref{eq:rcond}.
For our choice of $w$, it follows from the  derivations in the appendix that the RHS of \eqref{eq:rcond} becomes 
\begin{align*}
r < r_{\low} = \frac{1}{2} \frac{c}{1 - c} \ \text{where} \ c = \hat{f}(\emptyset) - \frac{\one^T A^+ A \one}{2^k}
\end{align*}
where $A$ is defined in Section \ref{sec:mainresult}. There is a slight technicality in directly applying Lemma \ref{lem:achlioptaslem} because $\phi_r$ might not be nonnegative for $\alpha < 1/2$; we discuss this in the appendix. Now using \eqref{eq:psiE[X]sq} and \eqref{eq:psiE[Xsq]}, and applying Lemma \ref{lem:achlioptaslem}, we can conclude that there exists $C > 0$ such that 
\begin{align*}
\Pr[C_f(n, rn) \ \text{is satisfiable}] \ge \E[X]^2/\E[X^2] \ge C 
\end{align*}
for sufficiently large $n$ and all $r < r_{\low}$.
\end{proof}
There remains a question of what $r_{\low}$ we can hope achieve using a second moment method proof where $X$ is defined as in \eqref{eq:X}. The following lemma provides some intuition for this: 
\begin{lemma}
\label{lem:maxr}
In order for the conditions of Lemma \ref{lem:achlioptaslem} to hold at $\alpha_{\max} = 1/2$ for $X$ in the form of \eqref{eq:X} and any choice of $w$ satisfying \eqref{eq:wlambdaf}, we require 
\begin{equation}
\label{eq:maxr}
r < \log 2 / \log \frac{1}{\hat{f}(\emptyset) - \frac{\one^T A^+ A \one}{2^k}} \le r_{\up} = \frac{\log 2}{\log \frac{1}{\hat{f}(\emptyset)}}
\end{equation}
\end{lemma}
 
The difference of $\one^T A^+ A \one/2^k$ in the lower logarithm compared to $r_{\up}$ can be viewed as a ``symmetrization penalty'' necessary for our proof to work. 
While Lemma \ref{lem:maxr} does not preclude applications of the second moment method that do not rely on Lemma \ref{lem:achlioptaslem}, consider what happens for $r$ that do not satisfy \eqref{eq:maxr}. For these $r$, the function $\psi_r(\alpha)$ must obtain a maximum at some $\alpha^* \in [0, 1], \alpha^* \ne 1/2$. If it also happens that $\phi_r(\alpha)$ is nonnegative and twice differentiable on $[0, 1]$, and $\psi_r''(\alpha^*) < 0$, then conditions of Lemma \ref{lem:achlioptaslem} hold, and applying it to $\alpha_{\max} = \alpha^*$ along with \eqref{eq:lemachlioptasres}, \eqref{eq:psiE[X]sq}, and \eqref{eq:psiE[Xsq]} will actually imply that 
\begin{align*}
\frac{\E[X]^2}{\E[X^2]} \le \frac{1}{B} \left(\frac{\psi_r(1/2)}{ \psi_r(\alpha^*)}\right)^n
\end{align*}
for some constant $B > 0$, which gives us an exponentially decreasing, and therefore trivial lower bound for the second moment method. Therefore, we believe that \eqref{eq:maxr} is near the best lower bound on $r_{\unsat}$ that we can achieve by applying the second moment method on $X$ in the form of \eqref{eq:X}. 

\section{Experimental Verification of Bounds}
\label{sec:experiments}
We empirically test our bounds with the goal of examining their tightness. For our constraint functions, we will use tribes functions. The tribes function takes the disjunction of $b$ groups of $a$ variables and evaluates to 1 or 0 based on whether the following formula is true: 
\begin{align*}
\tribes_{a, b}(x_1, \ldots, x_{ab}) = \vee_{i = 0}^{b - 1} \left(\wedge_{j = 1}^a x_{ia + j}\right)
\end{align*}
where $+1$ denotes true and $-1$ denotes false. For our experiments, we randomly generate CSP formulas based on $\tribes_{a, b}$. We use the Dimetheus\footnote{https://www.gableske.net/dimetheus} random CSP solver to solve these formulas, or report if no solution exists. We show our results in Figure \ref{fig:experiments}. As expected, our values for lower bounds $r_{\low}$ are looser than the upper bounds $r_{\up}$. 

\section{Conclusion}
Using Fourier analysis and the second moment method, we have shown general bounds on $m/n$, the ratio of constraints to variables; for $m/n$ below these bounds, there is constant probability that a random CSP is satisfiable. We demonstrate that our bounds are easily instantiated and can be applied to obtain novel estimates of the satisfiability threshold for many classes of CSPs. Our bounds depend on how easy it is to symmetrize solutions to the constraint function. We provide a heuristic argument to approximate the best possible lower bounds that our application of the second moment method can achieve; these bounds differ from upper bounds on the satisfiability threshold by a ``symmetrization penalty.'' Thus, an interesting direction of future research is to determine whether we can provide tighter upper bounds that account for symmetrization, or whether symmetrization terms are an artificial product of the second moment method.  

\section{Acknowledgements}
This work was supported by the Future of Life Institute
(grant 2016-158687) and by the National Science Foundation (grant 1649208).

\bibliographystyle{aaai}
\bibliography{refs}

\clearpage
\newpage
\section{Appendix: Optimizing $w$}
\label{sec:optimize}
Recall that we wish to choose $w$, our function used to weight solutions to the constraints in $C_f(n, rn)$ so that $\hat{w}(\emptyset)$ is as close as possible to $\hat{f}(\emptyset)$, where $f$ is our constraint function. We require that \eqref{eq:derivhalf0} and \eqref{eq:wlambdaf} hold, and we wish to maximize the bound in Lemma \ref{lem:rbound}. For convenience, we let $\hat{w} = (\hat{w}(S))_{S \subseteq [k]}$, the vector of $2^k$ Fourier coefficients of $w$. The bound in \eqref{eq:rcond} depends only on $\hat{w}(\emptyset)^2/ (\|\hat{w}\|_2^2 - \hat{w}(\emptyset)^2)$, so we can write the following optimization problem: 
\begin{equation}
\begin{split}
\text{Maximize} &\ \hat{w}(\emptyset)^2\\ 
\text{Subject to} &\ \| \hat{w} \|^2 = 1\\
& \hat{w}_{S : |S| = 1} = 0\\
& w(u) = \lambda(u)f(u) \ \text{for all} \ u \in \{-1, 1\}^k
\end{split}
\label{eq:optproblemw}
\end{equation}

The combination of the objective and the first constraint ensures that we maximize $\hat{w}(\emptyset)^2/(\|\hat{w}\|_2^2 - \hat{w}(\emptyset)^2)$. To see this, we note that we can scale $w$ so $\|w\|_2^2 = 1$ without changing the value of $\hat{w}(\emptyset)^2/(\|\hat{w}\|_2^2 - \hat{w}(\emptyset)^2)$. The second constraint ensures that \eqref{eq:derivhalf0} holds, and the third constraint is \eqref{eq:wlambdaf}. 
In this section, we will derive the solution to this optimization problem. 

Recalling that \eqref{eq:wlambdaf} asserts that $w(u) = \lambda (u) f(u)$, we will express the Fourier coefficients of $w$ in terms of $\lambda$ and $f$. For convenience, we will let $U = \{u \in \{-1, 1\}^k : f(u) = 1\}$. We note that from \eqref{eq:hatfS}, for any subset $S \subseteq [k]$
\begin{equation}
\begin{split}
\hat{w}(S) &= \frac{1}{2^k} \sum_{u \in \{-1, 1\}^k} \lambda(u) f(u)\chi_S(u)\\
&= \frac{1}{2^k} \sum_{u \in U} \lambda(u) \chi_S(u)
\end{split}
\label{eq:hatwlambda}
\end{equation}
We formulate this as a matrix-vector product. Overloading notation, we define the $2^k$-dimensional vector $\chi(u) = (\chi_S(u))_{S \subseteq [k]}$. Define the $2^k \times |U|$ matrix $M$ whose columns are the vectors $\chi(u)$ for all $u \in U$, i.e. the matrix whose columns are the Fourier basis functions evaluated at solutions of $f$. Finally, we let $\lambda$ be the $|U|$-dimensional vector whose entries correspond to $\lambda(u)$, arranged in the same order as the columns of $M$. Then from \eqref{eq:hatwlambda}, $M\lambda = 2^k\hat{w}$. Using this, we will transform the optimization problem in \eqref{eq:optproblemw} into a problem over $\lambda$. 

First, the condition that $\hat{w}_{S : |S| = 1} = 0$ becomes the condition that $A \lambda = 0$, where $A = M_{S : |S| = 1}$, the submatrix of $M$ consisting of all rows corresponding to size 1 subsets of $[k]$. Let $A^+$ denote the Moore-Penrose pseudoinverse of $A$. Since $\lambda$ is in the null space of $A$, we can write it as the projection of any arbitrary vector $\beta \in \mathbb{R}^{|U|}$ as follows: 
\begin{equation}
\label{eq:lambdasbeta}
\lambda = (I - A^+A) \beta
\end{equation}
where $I - A^+A$ is a symmetric, idempotent projection matrix \cite{barata2012moore}. The following example illustrates our notation.

\begin{example} [$3$-MAJORITY]
In the case of $3$-MAJORITY, we can instantiate $M$ as follows, up to permutation of the columns. The submatrix $A$ consists of the second to fourth rows of $M$. We label rows by which basis function they correspond to. 
\begin{align*}
M = \begin{array}{c c} 
\begin{array} {c}
\chi_{\emptyset} \\ 
\chi_{\{1\}}\\
\chi_{\{2\}}\\
\chi_{\{3\}}\\
\chi_{\{1, 2\}}\\
\chi_{\{1, 3\}}\\
\chi_{\{2, 3\}}\\
\chi_{\{1, 2, 3\}}\\
\end{array} &
\left[
\begin{array} {r r r r}
1 & 1 & 1 & 1\\
1 & 1 & 1 & -1\\
1 & 1 & -1 & 1\\
1 & -1 & 1 & 1\\
1 & 1 & -1 & -1\\
1 & -1 & 1 & -1\\
1 & -1 & -1 & 1\\
1 & -1 & -1 & -1\\
\end{array}
\right]
\end{array}
\end{align*}
\end{example}

Now we opt to write our constraint and objective in terms of $\beta$. We note that 
\begin{align*}
\|w\|_2^2 &= \frac{1}{2^{2k}}\lambda^T M^T M \lambda \\
&= \frac{1}{2^k} \lambda^T \lambda 
\end{align*}
because the columns of $M$ are orthogonal to each other. Plugging in our expression from \eqref{eq:lambdasbeta} and using the fact that $I - A^+A$ is symmetric and idempotent, we get 
\begin{align*}
\|w\|_2^2 &= \frac{1}{2^k} \beta^T (I - A^+A) \beta
\end{align*}
which gives us the transformed constraint 
\begin{equation}
\label{eq:betaconstraint}
\beta^T (I - A^+A) \beta = 2^k
\end{equation}
Next, we transform our objective. Since $\hat{w}(\emptyset) = \one^T \lambda/2^k$, plugging in \eqref{eq:lambdasbeta} and including the constraint \eqref{eq:betaconstraint}, we get the final optimization problem
\begin{equation}
\begin{split}
\text{Maximize} &\ \one^T (I - A^+A) \beta\\
\text{Subject to} &\ \beta^T (I - A^+A) \beta = 2^k
\end{split}
\label{eq:finalopt}
\end{equation}
The following lemma gives us the optimal solution to this problem: 
\begin{lemma}
\label{lem:optsolve}
The (not necessarily unique) optimal $\beta$ for this optimization problem occurs at 
\begin{align*}
\beta = \sqrt{\frac{2^k} {\one^T (I - A^+A)\one}} \cdot \one
\end{align*}
for which the corresponding value of $\hat{w}(\emptyset)$ is 
\begin{align*}
\hat{w}(\emptyset) =  \sqrt{\frac{\one^T(I - A^+A) \one}{2^k}}
\end{align*}
\end{lemma}
\begin{proof}
Since $I - A^+A$ is idempotent and symmetric, an equivalent optimization problem as \eqref{eq:finalopt} is 
\begin{align*}
\text{Maximize} &\ \one^T (I - A^+A)^T(I - A^+A) \beta \\
\text{Subject to} &\ \beta^T (I - A^+A)^T(I - A^+A) \beta = 2^k
\end{align*}
From this, it is clear that the direction of $(I - A^+A)\beta$ which maximizes the objective occurs when $(I - A^+ A)\beta = \rho (I - A^+ A) \one$, in which case we can choose $\beta = \rho \one$ for some $\rho \in \mathbb{R}$. Solving for $\rho$ gives us the desired result. 
\end{proof}
If we plug this value for $\hat{w}(\emptyset)$ into the bound in Lemma \ref{lem:rbound}, we will get the value of $r_{\low}$ described in Theorem \ref{thm:main}. 

Lemma \ref{lem:optsolve} does not only give us the maximum value of the bound in Lemma \ref{lem:rbound}; it also gives us a sense of the best $r_{\low}$ we can achieve using our proof techniques. We can use it to prove Lemma \ref{lem:maxr}: 

\begin{proof} [Proof of Lemma \ref{lem:maxr}]
First, recall that in order for Lemma \ref{lem:achlioptaslem} to hold at $\alpha_{\max} = 1/2$, we require \eqref{eq:derivhalf0} to be true. Now for any $w$, we also require $\psi_r(1/2) > \psi_r(1)$ if the lemma is to apply for $\alpha_{\max} = 1/2$. Expanding this condition, we get 
\begin{align*}
2g_w(1/2)^r &> g_w(1)^r\\
\implies \log 2 &> r \log \frac{g_w(1)}{g_w(1/2)}
\end{align*}
Since $g_w(1)/g_w(1/2) > 1$, the logarithm is positive so we can divide both sides and expand to get 
\begin{align*}
r &< \frac{\log 2}{\log \left(\frac{\sum_{S \subseteq [k]} \hat{w}(S)^2}{\hat{w}(\emptyset)^2}\right)}\\
&< \frac{\log 2}{\log \frac{2^k}{\one^T (I - A^+ A) \one}}
\end{align*}
We obtained the last line by maximizing the RHS over all $w$ that satisfy \eqref{eq:derivhalf0} and \eqref{eq:wlambdaf} and noting that we can apply our solution from Lemma \ref{lem:optsolve} here.
\end{proof}

\section{Appendix: Proofs and Loose Ends}
\subsection{Proofs}
\begin{proof} [Proof of Lemma \ref{lem:fourierc}]
In the case where the rows of $A$ are linearly independent, a property of the pseudoinverse tells us that $A^+ = A^T(A A^T)^{-1}$ \cite{barata2012moore}. Then by noting that $A \one = 2^k \hat{f}_{S : |S| = 1}$ and $A A^T = 2^k B$, 
\begin{align*}
c &= \hat{f}(\emptyset) - \frac{\one^T A^T (A A^T)^{-1} A \one}{2^k}\\
&= \hat{f}(\emptyset) - \hat{f}_{S : |S| = 1}^T B^{-1} \hat{f}_{S : |S| = 1}
\end{align*}
\end{proof}

\begin{proof} [Proof of Lemma \ref{lem:secmoment}]
As before, we will expand $\E[X^2]$ as follows: 
\begin{align}
\E[X^2] &= \sum_{\sigma, \tau \in \{-1, 1\}^n} \E\left[\prod_{j = 1}^{rn} w(\sigma_{I_j, s_j}) w(\tau_{I_j, s_j})\right] \notag \\
&= \sum_{\sigma, \tau \in \{-1, 1\}^n} \E[w(\sigma_{I, s})w(\tau_{I, s})]^{rn} \label{eq:EXsqinner}
\end{align}
Now we can compute $\E[w(\sigma_{I, s})\tau(\sigma_{I, s})]$ based on the number of indices $i$ where $\sigma_i = \tau_i$ which we will denote as $\alpha(\sigma, \tau) n$. First, we claim that for any $u, v \in \{-1, 1\}^k$,
\begin{equation}
\begin{split}
\Pr[\sigma_{I, s} = u, \tau_{I, s} = v] = \\ 
\frac{1}{2^k} \prod_{j = 1}^{k} \alpha(\sigma, \tau)^{\1(u_j = v_j)}(1 - \alpha(\sigma, \tau))^{\1(u_j \ne v_j)} 
\end{split}
\label{eq:prsigmatau}
\end{equation}
This holds since we can compute \eqref{eq:prsigmatau} by first computing the probability that $\sigma_{i_j} \tau_{i_j} = u_jv_j$ for all $i_j \in I$ and then the probability that $s_j\sigma_{i_j}$ and $s_j\tau_{i_j}$ actually match $u_j$ and $v_j$ given that the former is true. We can compute the probability of the former by noting that if $u_jv_j = 1$, then $i_j$ must correspond to one of the $\alpha(\sigma, \tau)$ indices where $\sigma_{i_j} = \tau_{i_j}$. The probability of the latter is simply $1/2^k$ since $s$ is chosen uniformly from $\{-1, 1\}^k$. Now note that  
\begin {align*}
\E[w(\sigma_{I, s})w(\tau_{I, s})] = &\sum_{u, v \in \{-1, 1\}^k} w(u)w(v) \cdot \\
&\Pr[\sigma_{I, s} = u, \tau_{I, s} = v]
\end{align*}
Expanding the probability using \eqref{eq:prsigmatau} and  writing $w$ in terms of its Fourier expansion, we get 
\begin{align}
\begin{split}
\E[w(\sigma_{I, s})\tau(\sigma_{I, s})] =\\
\sum_{u, v \in \{-1, 1\}^k} \left(\sum_{S \subseteq [k]} \hat{w}(S) \chi_S(u) \right) \left(\sum_{T \subseteq [k]} \hat{w}(T) \chi_T(v)\right) \cdot \\
\frac{1}{2^k} \prod_{i = 1}^k \alpha(\sigma, \tau)^{\1(u_i=v_i)}(1 - \alpha(\sigma, \tau))^{u_i \ne v_i} = 
\end{split}
\notag \\
\begin{split}
\frac{1}{2^k} \sum_{S, T \subseteq [k]} \hat{w}(S) \hat{w}(T) \left(\sum_{u, v \in \{-1, 1\}^k} \chi_S(u) \chi_T(v) \cdot \right.\\ 
\left. \prod_{i = 1}^k \alpha(\sigma, \tau)^{\1(u_i = v_i)} (1 - \alpha(\sigma, \tau))^{\1(u_i \ne v_i)}\right)
\end{split}
\label{eq:fourierexpand}
\end {align}
We can factorize the inner part of the summation as follows: 
\begin{align*}
\sum_{u, v \in \{-1, 1\}^k} \chi_S(u) \chi_T(v) \prod_{i = 1}^k \alpha(\sigma, \tau)^{\1(u_i = v_i)} \cdot \\ 
(1 - \alpha(\sigma, \tau))^{\1(u_i \ne v_i)} =\\ 
\prod_{i \in S \cap T} (4\alpha(\sigma, \tau) - 2) \prod_{i \in S \setminus T} 0 \prod_{i \in T \setminus S} 0 \prod_{i \in [k] \setminus (S \cup T)} 2\\
\end{align*}
Thus, terms where $S \ne T$ cancel to 0, so plugging the result back into \eqref{eq:fourierexpand} gives 
\begin{align*}
\E[w(\sigma_{I, s})\tau(\sigma_{I, s})] &= \sum_{S \subseteq [k]} (2 \alpha(\sigma, \tau) - 1)^{|S|} \hat{w}(S)^2\\ 
&= g_w(\alpha(\sigma, \tau))
\end{align*}
Plugging this expression back into \eqref{eq:EXsqinner}, we get that 
\begin{align*}
\E[X^2] = \sum_{\sigma, \tau \in \{-1, 1\}^n} g_w(\alpha(\sigma, \tau))
\end{align*}
Now for every $\sigma$ and $j = 0, \ldots, n$, there are exactly ${n \choose j}$ choices of $\tau$ for which $\alpha(\sigma, \tau) = j/n$. Thus, combining terms this way finally gives us the expression in \eqref{eq:E[Xsq]}.
\end{proof}

\begin{proof} [Proof of Lemma \ref{lem:alphagehalf}]
We have 
\begin{align*}
g_w(\alpha) - g_w(1 - \alpha) = \\ 
\sum_{S \subseteq [k]} \hat{w}(S)^2 \left[(2\alpha - 1)^{|S|} - (1 - 2\alpha)^{|S|} \right] = \\
\sum_{S \subseteq [k]} \hat{w}(S)^2(2\alpha - 1)^{|S|}(1 - (-1)^{|S|})
\end{align*}
Since $2\alpha - 1 \ge 0$ for $\alpha \ge 1/2$, all terms in the sum will be positive so $g_w(\alpha) \ge g_w(1 - \alpha)$.
\end{proof}

\begin{proof} [Proof of Lemma \ref{lem:rbound}]
We first show that $\psi_r(\alpha)$ takes its maximum at $\alpha = 1/2$. For all $\alpha \in [1/2, 1]$, we wish to prove
\begin{align*}
\frac{g_w(\alpha)^r}{\alpha^{\alpha}(1 - \alpha)^{1 - \alpha}} < 2g_w(1/2)^r = \psi_r(1/2)
\end{align*}
Rearranging and taking the logarithm of both sides gives
\begin{align*}
r \log \frac{g_w(\alpha)}{g_w(1/2)} < \log 2 + \alpha\log \alpha + (1 - \alpha)\log(1 - \alpha)
\end{align*}
Since $g_w(\alpha) > \hat{w}(\emptyset)^2 = g_w(1/2)$ for $\alpha > 1/2$, $\log \frac{g_w(\alpha)}{g_w(1/2)}$ is nonnegative so we can divide it from both sides. Thus, it suffices to have 
\begin{equation}
r < \frac{\log 2 + \alpha \log \alpha + (1 - \alpha) \log (1 - \alpha)}{\log \frac{g_w(\alpha)}{g_w(1/2)}} \label{eq:rlog}
\end{equation}
By using the rule $\log(1 + x) \le x$, we obtain  
\begin{align*}
\log \frac{g_w(\alpha)}{g_w(1/2)} &\le \frac{g_w(\alpha) - g_w(1/2)}{g_w(1/2)}\\
&\le \frac{(2\alpha - 1)^2\sum_{S : |S| \ge 2} \hat{w}(S)^2(2 \alpha - 1)^{|S| - 2}}{\hat{w}(\emptyset)^2}
\end{align*}
We got to the second line by expanding and applying the fact that \eqref{eq:derivhalf0} holds. 
Plugging this into \eqref{eq:rlog} shows that it is sufficient for 
\begin {equation}
\begin{split}
r <& \frac{\log 2 + \alpha \log \alpha + (1 - \alpha) \log (1 - \alpha)}{(2\alpha - 1)^2} \cdot \\  &\frac{\hat{w}(\emptyset)^2}{\sum_{S : |S| \ge 2} \hat{w}(S)^2 (2\alpha - 1)^{|S| - 2}}
\end{split}
\label{eq:rlogbound}
\end{equation}
We can check that the function 
\begin{align*}
h(\alpha) = \frac{\log 2 + \alpha \log \alpha + (1 - \alpha) \log (1 - \alpha)}{(2\alpha - 1)^2}
\end{align*}
is increasing on the interval $[1/2, 1]$, and in addition $\lim_{\alpha \rightarrow 1/2} h(\alpha) = 1/2$. Meanwhile, 
\begin{align*}
\min_{\alpha \in [1/2, 1]}\frac{\hat{w}(1/2)^2}{\sum_{S : |S| \ge 2} \hat{w}(S)^2 (2\alpha - 1)^{|S| - 2}} =\\ 
\frac{\hat{w}(\emptyset)^2}{\sum_{S : |S| \ge 2} \hat{w}(S)^2}  
\end{align*}
at $\alpha = 1$. Therefore, for all $\alpha \in [1/2, 1]$, the RHS of \eqref{eq:rlogbound} is lower bounded by 
\begin{align*}
\frac{1}{2}\frac{\hat{w}(\emptyset)^2}{\sum_{S : |S| \ge 2} \hat{w}(S)^2}  
\end{align*}
which is precisely our condition on $r$ in \eqref{eq:rcond}. Thus, if \eqref{eq:rcond} holds,  
then \eqref{eq:rlogbound} must hold for $\alpha \in [1/2, 1]$, which implies that $\psi_r(\alpha)$ is maximized at $\alpha = 1/2$. Now we check that $\psi_r''(1/2) < 0$. We can first compute the first derivative as 
\begin{align*}
\psi_r'(\alpha) =\\ 
\frac{g_w(\alpha)^{r - 1}(rg_w'(\alpha) + g_w(\alpha)(\log(1 - \alpha) - \log(\alpha))}{\alpha^{\alpha}(1 - \alpha)^{1 - \alpha}}
\end{align*}
Since $g_w'(1/2) = 0$ from \eqref{eq:derivhalf0}, $\psi''(1/2) < 0$ if the derivative of 
\begin{align*}
rg_w'(\alpha) + g_w(\alpha)(\log(1 - \alpha) - \log(\alpha))
\end{align*} 
is negative. We can compute this derivative as 
\begin{align*}
rg_w''(\alpha) + g_w'(\alpha)\log \frac{1 - \alpha}{\alpha} - g_w(\alpha)\left(\frac{1}{\alpha} + \frac{1}{1 - \alpha}\right)
\end{align*}
Plugging in $\alpha = 1/2$, the expression simplifies to 
\begin{align*}
rg_w''(1/2) - 4g_w(1/2)
\end{align*}
Thus, if \eqref{eq:rcond} holds, 
\begin{align*}
r < \frac{\hat{w}(\emptyset)^2}{2\sum_{S : |S| = 2} \hat{w}(S)^2} = \frac{4g_w(1/2)}{g_w''(1/2)}
\end{align*}
which means that $\psi_r''(1/2)$ is also satisfied. 
\end{proof}
Using the lemmas developed earlier regarding the moments of $X$, we can complete the proof of Proposition \ref{prop:upper}. 
\begin{proof} [Proof of Proposition \ref{prop:upper}]
We let $X$ be the number of solutions to $C_f(n, rn)$. This corresponds to the choice of $w = f$. Then Lemma \ref{lem:firstmomentsq} tells us that 
\begin{align*}
\E[X] = (2 \hat{f}(\emptyset)^r)^n
\end{align*}
With 
\begin{align*}
r > r_{\up} = \frac{\log 2}{\log 1/\hat{f}(\emptyset)}
\end{align*}
we get $2 \hat{f}(\emptyset)^r < 1$. Therefore, $\lim_{n \rightarrow \infty} \E[X] = 0$. By Markov's inequality, 
\begin{align*}
\lim_{n \rightarrow \infty} \Pr[X \ge 1] \le \lim_{n \rightarrow \infty} \E[X] = 0
\end{align*}
as desired.
\end{proof}
Next, we justify our application of Lemma \ref{lem:achlioptaslem} from Section \ref{sec:mainproof} in greater detail. 

\begin{proof} [Detailed application of Lemma \ref{lem:achlioptaslem}]
We see that we cannot apply Lemma \ref{lem:achlioptaslem} directly to $\phi_r(\alpha) = g_w(\alpha)^r$, as $\phi_r(\alpha)$ may be negative for $\alpha < 1/2$ if $w$ can take on negative values. Instead, we define 
\begin{align*}
\phi_r^*(\alpha) = 
\begin{cases}
g_w(\alpha)^r \ &\text{if} \ \alpha \ge 1/2\\
g_w(1 - \alpha)^r \ &\text{if} \ \alpha < 1/2
\end{cases}
\end{align*}
The idea is to apply Lemma \ref{lem:achlioptaslem} to $\phi_r^*(\alpha)$ instead. 
The benefit of doing this is that $\phi_r^*(\alpha) > 0$ for all $\alpha \in [0, 1]$, as $g_w(\alpha) > 0$ if $\alpha \ge 1/2$. We will also check that $\phi_r^*(\alpha)$ is twice differentiable on $[0, 1]$.  For $\alpha \ne 1/2$, this is clear. Now we note that 
\begin{align*}
\frac{dg_w(1 - \alpha)^r}{d\alpha} = -r g_w(1 - \alpha)^{r - 1} g_w'(1 - \alpha)
\end{align*}
At $\alpha = 1/2$, this evaluates to $0$ because of \eqref{eq:derivhalf0}. Likewise, the derivative of $g_w(\alpha)^r$ at $\alpha = 1/2$ is also $0$. Thus, $\phi_r^*(\alpha)$ is first-order differentiable on $[0, 1]$. To show that the second derivative exists at $\alpha = 1/2$, we note that 
\begin{align*}
\frac{d^2g_w(1 - \alpha)^r}{d\alpha^2} = &r(r - 1) g_w(1 - \alpha)^{r - 2} \cdot g_w'(1 - \alpha)^2 + \\
& rg_w(1 - \alpha)^{r - 1} \cdot g_w''(1 - \alpha)
\end{align*}
If we compute the same expression for the second derivative of $g_w(\alpha)^r$, we see that the two expressions have identical terms at $\alpha = 1/2$. Thus, second derivatives match, so $\phi_r^*(\alpha)$ is twice-differentiable on $[0, 1]$. 

Now we can apply Lemma \ref{lem:achlioptaslem} to $\phi_r^*$. Define 
\begin{align*}
\psi_r^* = \frac{\phi_r^*}{\alpha^{\alpha}(1 - \alpha)^{1 - \alpha}}
\end{align*}
Since $\alpha^{\alpha}(1 - \alpha)^{1 - \alpha}$ is symmetric around $1/2$, Lemma \ref{lem:rbound} implies that for $r$ satisfying \eqref{eq:rcond}, $\psi_r^*(1/2) > \psi_r^*(\alpha)$ for all $\alpha \in [0, 1]$ where $\alpha \ne 1/2$. Furthermore, since $(\psi_r^*)''(1/2) = \psi_r''(1/2) < 0$ by symmetry around $1/2$, the conditions for Lemma \ref{lem:achlioptaslem} are satisfied so 
\begin{equation}
\label{eq:achlioptaslemoutcome}
\frac{\E[X]^2}{2^n\sum_{j = 0}^n {n \choose j} \phi_r^*(j/n)} > C
\end{equation}
for some constant $C > 0$ and sufficiently large $n$. Finally, we note that from \eqref{eq:E[Xsq]}, 
\begin{align*}
\E[X^2] &= 2^n \sum_{j = 0}^n {n \choose j} \phi_r(j/n)\\
&\le 2^n \sum_{j = 0}^n {n \choose j} \phi_r^*(j/n)
\end{align*}
because $\phi_r(\alpha) \le \phi_r^*(\alpha)$ for $\alpha < 1/2$, as a consequence of Lemma \ref{lem:alphagehalf}. Plugging this result into \eqref{eq:achlioptaslemoutcome} finally gives us 
\begin{align*}
\frac{\E[X]^2}{\E[X^2]} \ge \frac{\E[X]^2}{2^n\sum_{j = 0}^n {n \choose j} \phi_r^*(j/n)} > C
\end{align*}
\end{proof}
\subsection{Sampling With vs. Without Replacement}
We will provide an explanation for why we can assume that we can sample constraints and constraint indices with replacement when proving our main theorem. This explanation is due to Section 3 of \cite{achlioptas2004threshold} for $k$-SAT, and applies directly to our more general setting. 

We call a constraint $f_{I, s}$ improper if $I$ contains repeated variables or $f_{I, s}$ is itself repeated. The probability that a constraint contains repeated variables is bounded above by $k^2/n$, so with high probability there are $o(n)$ constraints with repeated variables. Likewise, with high probability there are $o(n)$ repeated clauses, so there are $o(n)$ improper constraints with high probability. Furthermore, the distribution over proper constraints remains uniform in the with-replacement setting. Thus, if $C_f(n, rn)$ is satisfiable with constant probability in the with-replacement setting for $m = rn$, then it will be satisfiable with constant probability in the without-replacement setting for $m = rn - o(n)$ constraints. Since we only subtract a $o(n)$ factor, we retain the same ratio $r$.

\end{document}